\documentclass[lettersize,journal]{IEEEtran}
\usepackage{amsmath,amsfonts}
\usepackage{algorithmic}
\usepackage{cite}
\usepackage{array}
\usepackage[caption=false,font=normalsize,labelfont=sf,textfont=sf]{subfig}
\usepackage{textcomp}
\usepackage{stfloats}
\usepackage{url}
\usepackage{verbatim}
\usepackage{mwe}
\usepackage{algorithm}
\usepackage{amssymb}
\usepackage{stfloats}
\usepackage{cuted}
\usepackage{booktabs}
\usepackage{makecell}
\usepackage{amssymb}
\usepackage{graphicx}
\usepackage{subfig}

\usepackage{amsthm}

\usepackage{caption}
\newtheorem{thm}{\bf Theorem}
\newtheorem{rmk}{\bf Remark}
\DeclareMathOperator*{\argmin}{arg\,min}

\def\BibTeX{{\rm B\kern-.05em{\sc i\kern-.025em b}\kern-.08em
    T\kern-.1667em\lower.7ex\hbox{E}\kern-.125emX}}
\usepackage{balance}
\begin{document}
\title{A Low-Complexity Joint Fractional Delay and Doppler Frequency Estimator for AFDM-Enabled Vehicular LEO-ICAN Systems}
\author{Zhenyu Chen, Ke Xiao, Xiaomei Tang, Jing Lei, Muzi Yuan and Guangfu Sun
\thanks{This work is supported by the National Natural Science Foundation of China (Grant No. 62531025). (\textit{Corresponding author: Xiaomei Tang})
	
Z. Chen, K. Xiao, X. Tang, J. Lei, M. Yuan and G. Sun are with the College of Electronic Science and Technology, National University of Defense Technology, Changsha 410073, China. email: (chenzhenyu20, xiaoke3\underline{ }3, tangxiaomei, leijing, ymz)@nudt.edu.cn, gfsunmail@163.com.

The authors are grateful to Prof. Yanqun Tang, Dr. Yu Zhou and Haoran Yin from Sun Yat-sen University.}}

%\markboth{Journal of \LaTeX\ Class Files,~Vol.~18, No.~9, September~2020}%
%{How to Use the IEEEtran \LaTeX \ Templates}

\maketitle

\begin{abstract}
Low-Earth-orbit (LEO) satellites and vehicle-to-everything (V2X) networks are driving integrated communication and navigation (ICAN) toward next-generation intelligent transportation. Affine frequency division multiplexing (AFDM) is a promising waveform for high-mobility LEO scenarios owing to its Doppler robustness, simple modulation, and low pilot overhead. However, applying existing high-accuracy AFDM fractional delay-Doppler estimators to LEO-ICAN entails substantial search or inference complexity, while the spectrum-wrapping-induced envelope structure in line-of-sight (LOS)-dominated channels remains underexploited. This paper analyzes and exploits the spectrum-wrapping-induced envelope structure of the fractional AFDM response, and proposes a low-complexity joint estimator that combines minimum-entropy fractional Doppler estimation with closed-form fractional delay estimation. Simulation results show that the proposed estimator approaches the root Cram\'er--Rao lower bound (RCRLB) and achieves root-mean-square error (RMSE) performance comparable to that of matched filtering (MF), matched filtering with generalized Fibonacci search (MF-GFS), and off-grid sparse Bayesian learning (OG-SBL), while requiring substantially lower computational complexity and runtime. This favorable accuracy-complexity profile highlights the potential of the proposed estimator for real-time ICAN processing in high-mobility LEO-assisted vehicular networks.
\end{abstract}


\section{Introduction}
\IEEEPARstart{T}{raditional} communication and navigation systems operate in isolated modes, leading to redundant infrastructure and fragmented data resources. Integrated communication and navigation (ICAN) overcomes these limitations through resource sharing and information fusion, which is essential for connected and autonomous vehicles that simultaneously demand lane-level positioning and high-throughput data transmission~\cite{11016809,10726910}.

\begin{figure}[htp]
	\centering
	\includegraphics[width=3.5in]{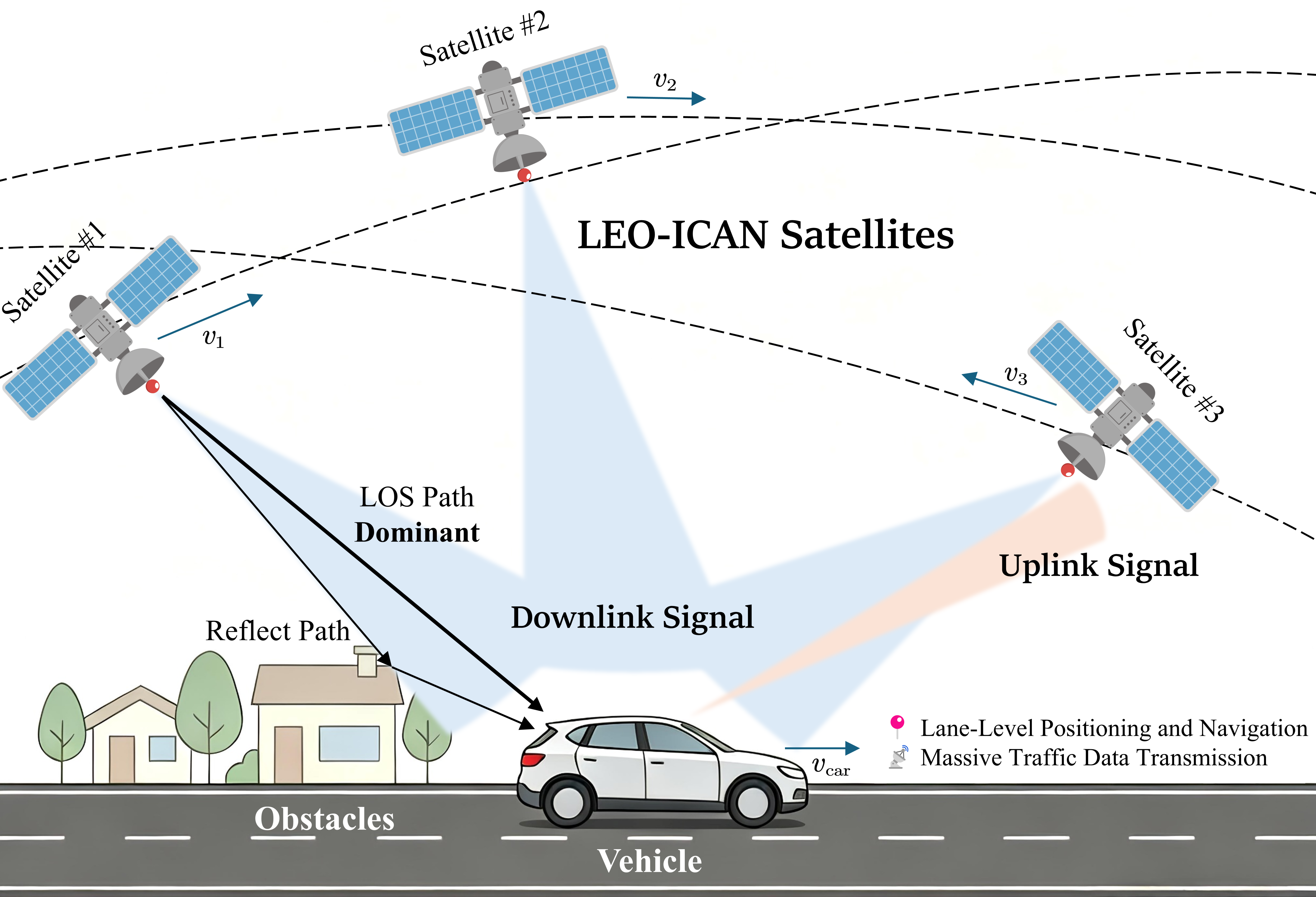}
	\caption{Vehicular LEO-ICAN system scenario, enabling simultaneous lane-level positioning and massive data transmission for autonomous vehicles.}
	\label{scene}
\end{figure}

As next-generation mobile networks advance toward air-space-ground integration, low-Earth-orbit (LEO) satellites have become an ideal ICAN platform for ubiquitous coverage. However, the severe Doppler effect in high-mobility LEO links poses a major challenge to conventional orthogonal frequency division multiplexing (OFDM)~\cite{hadaniOrthogonalTimeFrequency2017a}. Affine frequency division multiplexing (AFDM), built upon the discrete affine Fourier transform (DAFT), has therefore emerged as a promising alternative. Compared with orthogonal time frequency space (OTFS), AFDM offers a simpler modulation structure and lower pilot overhead, making it a promising waveform for LEO-ICAN systems~\cite{bemaniAffineFrequencyDivision2023a,bemaniAffineFrequencyDivision2023}.

Channel estimation is indispensable in ICAN signal processing: accurate delay and Doppler estimation enables both high-precision vehicle positioning using signals from three or more LEO satellites and reliable data transmission in high-mobility links, as illustrated in Fig.~\ref{scene}. Existing fractional delay--Doppler estimators for AFDM mainly fall into two categories: matched-filtering (MF) schemes with high accuracy but heavy complexity, primarily improved via fast-search strategies~\cite{bemaniIntegratedSensingCommunications2024,liMatchedFilteringbasedChannel2026a,zhuLowComplexityAFDMbased2026,luoNovelAngledelaydopplerEstimation2025b}; and sparse Bayesian learning (SBL) schemes that excel in sparse multipath channels~\cite{yangAFDMOffgridChannel2024,liLowcomplexityChannelEstimation2026}. Although fractional delay--Doppler channel modeling for AFDM has been established~\cite{zhuLowComplexityAFDMbased2026}, directly applying existing high-accuracy estimators to LEO-ICAN entails substantial search or inference complexity, with the spectrum-wrapping-induced envelope structure under LOS-dominated channels remaining underexploited for scenario-tailored complexity reduction.

Motivated by these gaps, this paper makes two contributions: \emph{First}, we analyze the envelope of the fractional-delay AFDM input--output relationship induced by spectrum wrapping, and provide an approximate closed-form expression for the nonzero responses after fractional Doppler compensation. \emph{Second}, targeting the line-of-sight (LOS)-dominated LEO satellite-ground links, we propose an entropy-driven Doppler estimation-compensation scheme combined with a closed-form delay estimator, achieving markedly lower complexity with comparable accuracy to baselines.

\section{AFDM Basics}
Fig.~\ref{diagram} illustrates the AFDM-enabled ICAN architecture, where the inverse DAFT (IDAFT) generates discrete-time AFDM samples, which are first appended with a chirp periodic prefix (CPP) and then converted into a continuous-time waveform through digital-to-analog conversion (DAC). The resulting waveform exhibits the characteristic time--frequency spectrum wrapping behavior~\cite{yinAmbiguityFunctionAnalysis2025b}. At the receiver, the estimated channel parameters are used for both data detection and navigation acquisition/tracking.

\begin{figure}[htp]
	\centering
	\includegraphics[width=3.5in]{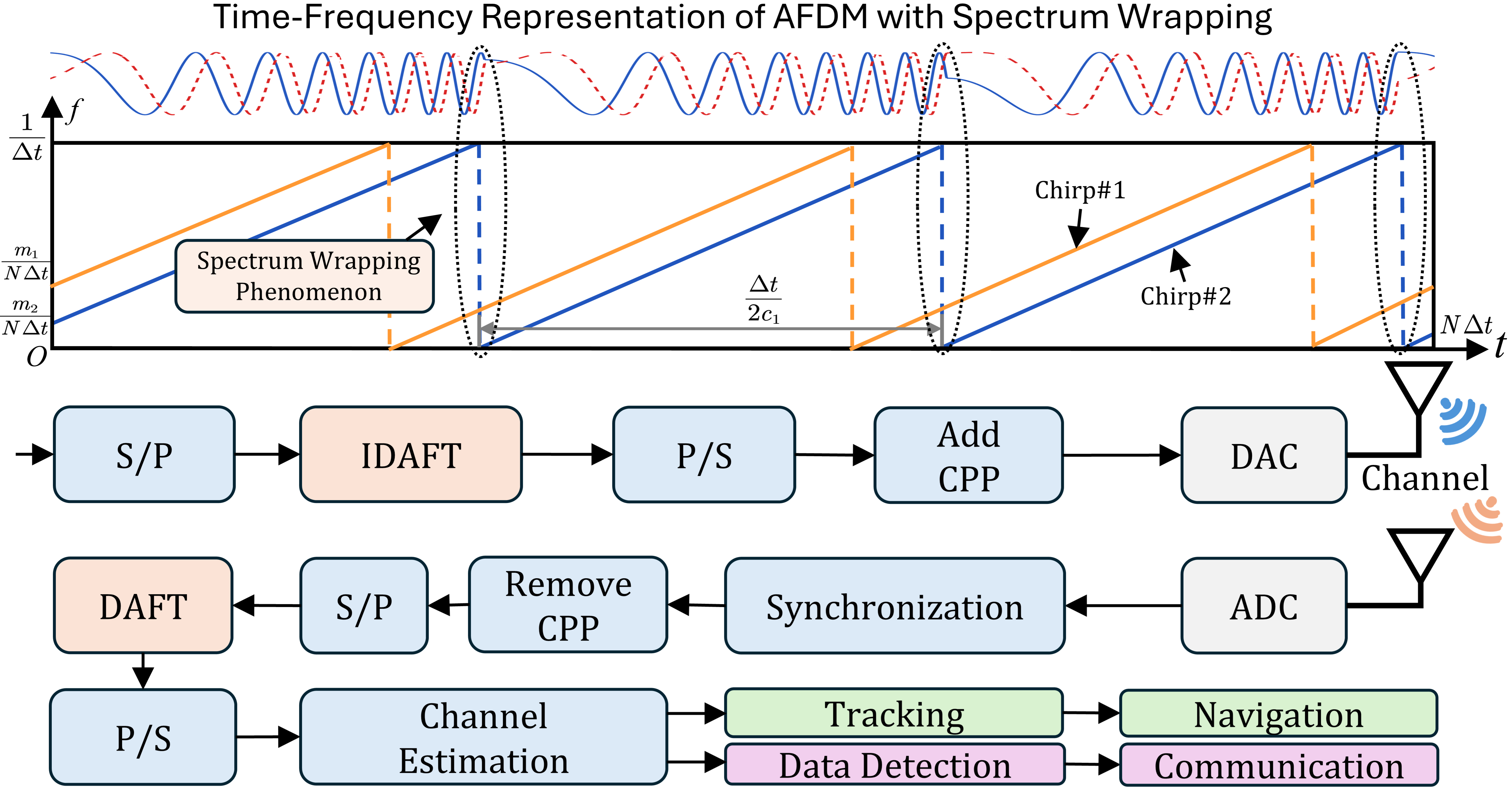}
	\caption{Block diagram of AFDM-enabled ICAN system with time--frequency representation}
	\label{diagram}
\end{figure}

AFDM modulation is expressed as \cite{bemaniAffineFrequencyDivision2023}:
\begin{equation}
	\label{afdm_mod}
	s[n]=\sum_{m^{\prime}=0}^{N-1}{x[m^{\prime}]\phi _n(m^{\prime}),\quad n}=0,\cdots ,N-1,
\end{equation}
where $N$ is the number of subcarriers, and $ \phi _n\left( m^{\prime} \right) =\frac{1}{\sqrt{N}} e^{\imath2\pi \left( c_1n^2+c_2m^{\prime 2}+nm^{\prime}/N \right)} $ denotes the AFDM subcarrier basis function, with $m^{\prime}$ being the index of the transmitter subcarriers. To achieve full diversity gain, the chirp parameter $c_1$ is set as $c_1 = (2 k_{\text{max}}+1)/2N= C/2N$, where $k_{\text{max}}$ is the maximum normalized Doppler frequency, and $C$ is the total number of chirp segments. The parameter $c_2$ can be any irrational number, or a rational number far smaller than $1/(2N)$.

To mitigate multipath fading, a CPP is appended to the time-domain AFDM signal during modulation, defined as:
\begin{equation}
	\label{CPP}
	s[n]=s[N+n]e^{-\imath 2\pi c_1(N^2+2Nn)},\quad n=-n_{\mathrm{cp}},\cdots ,-1.
\end{equation}

Owing to the inherent spectrum wrapping feature of AFDM, phase compensation is required after DAC to ensure the waveform conforms to its standard time--frequency representation. The corresponding continuous-time AFDM signal is written as\cite{bemaniAffineFrequencyDivision2023}:

\begin{equation}
	\label{continous_time_signal}
	\sqrt{\Delta t}s(t)
	=
	\sum_{m^{\prime}=0}^{N-1}
	x[m^{\prime}]
	e^{\imath 2\pi \Psi_{m',q}(t)},
	\quad
	t_{m^{\prime},q}\le t<t_{m^{\prime},q+1}.
\end{equation}
where $	\Psi_{m',q}(t)
= c_2 m^{\prime 2}
+ c_{1}^{\prime} t^2
+ \frac{m^{\prime}}{T}t
+ \frac{q}{2c_1}\left(q-\frac{m^{\prime}}{N}\right)
- \frac{q}{\Delta t}t
+ \varphi_{m^{\prime},q}$, $\Delta t$ is the sampling interval, $c_{1}^{\prime}=c_1/\Delta t^2$ is the continuous-time chirp rate, $q=1,2,\dots,C$ is the chirp segment index, and the constant phase term $\varphi _{m^{\prime},q}$ is defined as $\varphi _{m^{\prime},q}=\lfloor \frac{q}{2c_1}(q-\frac{m^{\prime}}{N}) \rfloor -\frac{q}{2c_1}(q-\frac{m^{\prime}}{N})$ to ensure the sampled continuous-time signal can be recovered to the discrete form in (\ref{afdm_mod}). The piecewise time boundary $t_{m^{\prime},q}$ is given by:
\begin{equation}
	\label{t}
	\begin{split}
		t_{m^{\prime},q}=\begin{cases}
			0,&		q=0\\
			&		\\
			\frac{(N-m^{\prime})}{2Nc_1}\Delta t+\frac{q-1}{2c_1}\Delta t,&		1\le q\le C\\
		\end{cases}.
	\end{split}
\end{equation}

After traversing a time--frequency doubly selective channel, the received continuous-time signal is given by:
\begin{equation}
	\label{doubly selective channel}
	r(t)=\iint{h(\tau ,\nu )s(t}-\tau )e^{\imath 2\pi \nu t}\mathrm{d}\tau \mathrm{d}\nu +w\left( t \right),
\end{equation}
where $w(t)$ is additive complex white Gaussian noise following $\mathcal{C}\mathcal{N}(0,\sigma^2)$. In ICAN scenarios, the line-of-sight (LOS) component dominates the channel \cite{weiTimeofarrivalEstimationIntegrated2023a}, so the channel impulse response (CIR) is simplified to $h(\tau,\nu) = h\delta(\nu-\nu_0)\delta(\tau-\tau_0)$, where $h$, $\nu_0$, and $\tau_0$ denote the channel gain, Doppler frequency, and propagation delay, respectively.

\begin{figure*}[htb] 
	\centering

	\begin{equation}	  
		\label{received_signal}  
		r\left[ n \right] =\frac{1}{\sqrt{N}}\sum_{m^{\prime}=0}^{N-1}{hx[m^{\prime}]e^{\imath 2\pi (c_1(n-(l+\iota ))^2+c_2m^{\prime 2}+\frac{n-(l+\iota ))m^{\prime}}{N}-\frac{k+\kappa }{N}n)}e^{\imath 2\pi (\sum_{q=0}^C{q\iota I_{\mathcal{L} _{m^{\prime},q}}((n}-(l+\iota))_N))}}+w[n]
	\end{equation}
		\hrulefill
\end{figure*}

We first define $L = l + \iota = \tau_0/\Delta t$ and $K = k + \kappa = \nu_0/\Delta f$, where $\Delta f$ is the subcarrier spacing, $l$ and $k$ are the integer parts of the normalized delay and Doppler frequency, respectively, and $\iota$ and $\kappa$ are the corresponding fractional parts within (-0.5,0.5). After analog-to-digital conversion (ADC), the discrete received signal is given by (\ref{received_signal}) at the bottom of the next page, where $(\cdot)_N$ denotes the modulo-$N$ operation, and the indicator function $I_{\mathcal{L}_{m^{\prime},q}}(n)$ is defined as
\begin{equation}
	\label{indicator_fun}
	\begin{split}
		I_{\mathcal{L}_{m^{\prime},q}}(n) = 
		\begin{cases} 
			1 & n \in \mathcal{L}_{m^{\prime},q} \triangleq \{n_{m^{\prime},q}+1, \cdots, n_{m^\prime,q+1}\}, \\
			0 & n \notin \mathcal{L}_{m^{\prime},q}. 
		\end{cases}
	\end{split}
\end{equation}
where $ n_{m^{\prime},q}\triangleq \lfloor \frac{t_{m^\prime,q}}{\varDelta t} \rfloor $.

AFDM demodulation is implemented via DAFT on the received discrete signal, expressed as
\begin{equation}
	\label{y_m}
	y[m] =\sum_{n=0}^{N-1}{r[n]   \phi _{n}^{*}\left( m \right)} ,
\end{equation}
where $(\cdot)^*$ denotes the complex conjugate operation, $m$ denotes the index of the subcarrier at the receiver and the noise term $\tilde{w}[m]$ still follows the complex Gaussian distribution $\mathcal{C}\mathcal{N}(0,\sigma^2)$.

\section{Joint Delay and Doppler Frequency Estimator}
This section first derives the input--output relationship of AFDM, and on this basis, presents a joint fractional delay and Doppler frequency estimation algorithm.
\subsection{Input--output Relationship of AFDM}
Equation (\ref{y_m}) can be expressed in matrix form as follows:
\begin{equation}
	\label{matrix_form_ym}
	\mathbf{y}=\mathbf{H}_{\mathrm{eff}}\mathbf{x}+\mathbf{w},
\end{equation}
where $\mathbf{H}_{\mathrm{eff}}$ is the effective channel matrix describing the channel properties in the DAFT domain.
\begin{figure*}[htb] 
	\centering
	\begin{equation}
		\label{H_p}
		\begin{split}
			\mathbf{H}_{\text{eff}}[m,m^{\prime}]=\frac{1}{N} h e^{\imath 2\pi (c_1(l+\iota )^2-c_2(m^2-m^{\prime 2})-(l+\iota )m^{\prime}/N)}\sum_{n=0}^{N-1}{e^{\imath \frac{2\pi}{N}(m^{\prime}-(m+l^{\mathrm{eq}}))n}e^{\imath 2\pi \iota (\sum_{q=0}^C{qI_{\mathcal{L} _{m^{\prime},q}}((n}-(l+\iota ))_N))}}
		\end{split}
	\end{equation}
		\hrulefill
\end{figure*}
The element in the $m$-th row and \(m^{\prime}\)-th column of \(\mathbf{H}_{\text{eff}}\) is expressed as (\ref{H_p}) above, where $l^{\text{eq}} \triangleq (K+CL)_N$ is the equivalent delay.

The summation term is the core part that needs to be analyzed, so it is separately defined as \(\mathcal{F}(m,m^{\prime})\):
\begin{equation}
	\label{F_p}
	\mathcal{F} (m,m^{\prime})=\sum_{n=0}^{N-1}{e^{\imath \frac{2\pi}{N}(m^{\prime}-(m+l^{\mathrm{eq}}))n}e^{\imath 2\pi \iota (\sum\limits_{q=0}^{C}{qI_{\mathcal{L} _{m^{\prime},q}}((n}-(l+\iota))_N))}}.
\end{equation}

Because of the spectrum wrapping phenomenon, the expression of \(\mathcal{F}(m, m')\) can be regarded as the superposition of several piecewise Fourier bases related to $q$. When $N$ is sufficiently large, the magnitude of \(\mathcal{F}(m,m')\) exhibits coherent characteristics, as demonstrated in Theorem 1:
\begin{thm}\label{thm1}
	When $N\gg \max\{C,l_{\max},M_C(m')\}$ and $|\Delta|\ll N$ in the main-lobe region, the magnitude of $\mathcal{F}(m,m')$ can be approximated as
	\begin{equation}
		\label{magnitude}
		\left|\mathcal{F}(m,m')\right|
		\approx
		\frac{N}{C}\left|\frac{\sin\!\left(\pi(C\iota+\Delta)\right)}{\sin\!\left(\tfrac{\pi}{C}(C\iota+\Delta)\right)}\right|
		\times\left|\mathrm{sinc}\!\left(\dfrac{\Delta}{C}\right)\right|,
	\end{equation}
	where $\Delta=m'-(m+l^{\mathrm{eq}})$, $\mathrm{sinc}(x)=\sin(\pi x)/(\pi x)$, $M_C(m')=\max\{-\lfloor\iota-m'/C\rfloor,0\}$.
\end{thm}
\begin{proof}[Proof of Theorem \ref{thm1}]
	See Appendix A.
\end{proof}
\begin{rmk}\label{rmk1}
	As shown in (\ref{magnitude}), the magnitude $\left|\mathcal{F}(m,m')\right|$ is composed of two distinct components. The first component, $\varUpsilon(m,m')=\left|\frac{\sin(\pi(C\iota+\Delta))}{\sin(\frac{\pi}{C}(C\iota+\Delta))}\right|$, is a Dirichlet kernel with respect to $m$ that exhibits a periodicity of $C$. The second component, $\varTheta(m,m')=\left|\mathrm{sinc}(\frac{\Delta}{C})\right|$, acts as a sinc-shaped envelope with respect to $m$.
	
	These properties are illustrated in Fig.~\ref{fig3}. The figure shows that $\Upsilon(m,m')$ depends exclusively on the fractional Doppler shift $\kappa$, since the fractional delay $\iota$ is eliminated from the Dirichlet kernel variables through its cancellation with the corresponding terms in $C\iota$ and $l^{\mathrm{eq}}$, leading to energy leakage. By contrast, $\Theta(m,m')$ is governed by both $\kappa$ and $\iota$, whose combined effect shifts the peak position. Therefore, $\kappa$ should be estimated and compensated before the estimation of $\iota$. Note that $M_C(m')\ll N$ does not hold for arbitrary $m'$ (worst case $M_C(m')\sim N/C$), but fortunately, $m'$ is a controllable pilot parameter, and setting $m'=0$ yields $M_C(0)\le 1$, thereby supporting the approximation in~(\ref{magnitude}).
\end{rmk}

\begin{figure}[!h]
	\centering
	\includegraphics[width=3in]{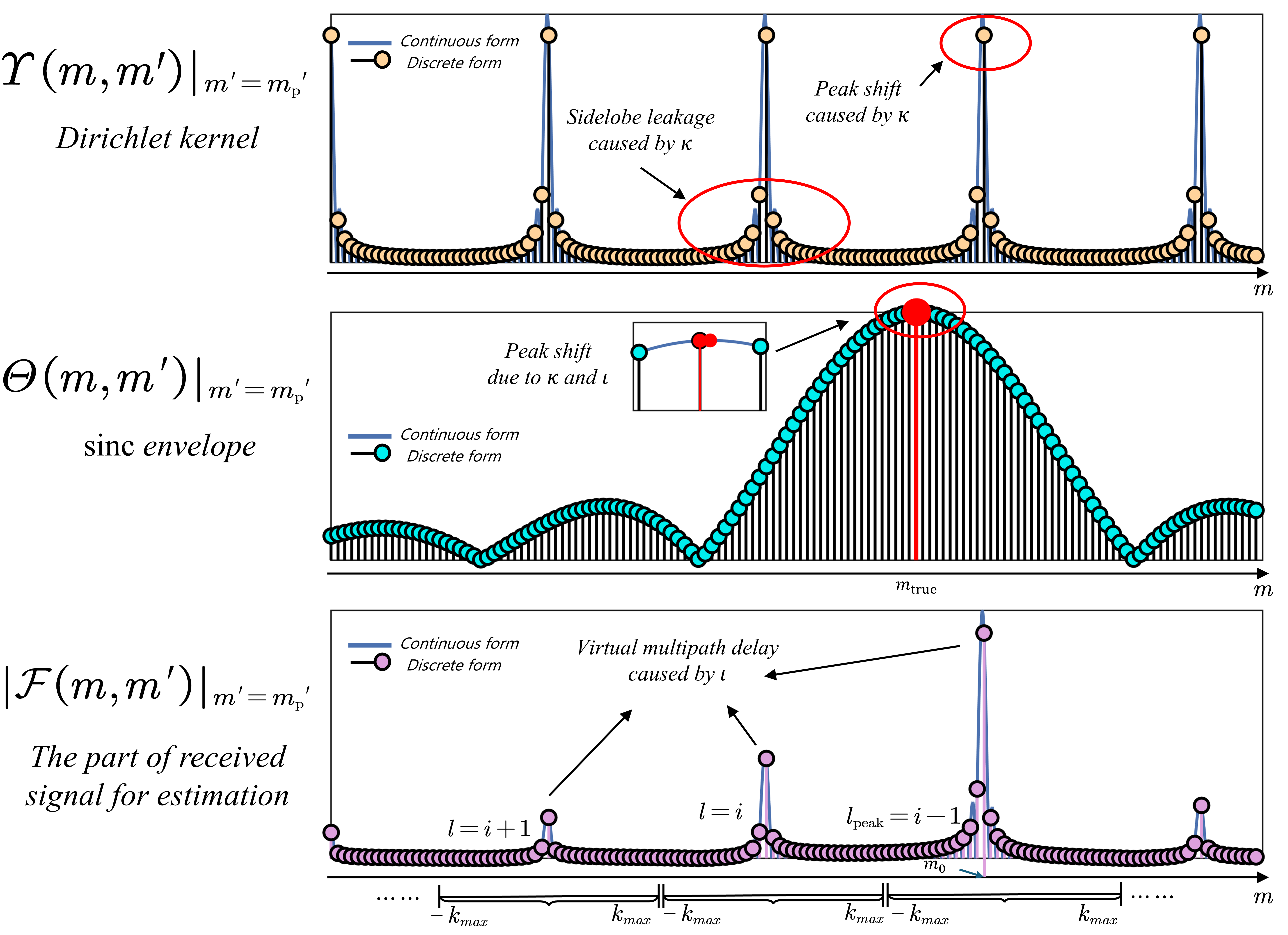}
	\caption{Schematic diagram of \(\Upsilon(m, m')\), \(\Theta(m, m')\), and \(|\mathcal{F}(m, m')|\) under $m$ index (with $m^{\prime}=0$).}
	\label{fig3}
	\vspace{-10pt}
\end{figure}

\subsection{Estimation Method}
The joint fractional delay and Doppler frequency estimator relies on the embedded single-pilot-aided (SPA) scheme [12], where the estimation portion of the received symbol \(y[m]\) corresponds to \(\mathbf{H}_{\rm eff}[m,m']|_{m'=m_{\rm p}'}\), and \(m_{\rm p}'\) denotes the pilot location in \(x[m']\).
The pilot symbol is placed at the first position of \(x[m']\), with \(Q\) guard symbols reserved on both sides in the circular index sense:
\begin{equation}	
	\label{xn}  
	x\left[ m^{\prime} \right] =\begin{cases}
		x_{\text{pilot}}, m^{\prime}=0\\
		0,     1\leqslant m^{\prime}\leqslant Q \ \text{or}\  N-Q+1\leqslant m^{\prime}\leqslant N-1\\
		x_{\text{data}},  Q+1\leqslant m^{\prime}\leqslant N-Q\\
	\end{cases},
\end{equation}
where \(Q = 2 l_{\text{max }} k_{\text{max }}+2 k_{\text{max }}+l_{\text{max }}\) and $l_{\text{max }}$ is the maximum normalized delay. The integer parts $\hat{l}$ and $\hat{k}$ can be  obtained by SPA, so we only consider the fractional parts.
\subsubsection{Minimum-Entropy-Based Fractional Doppler Frequency Search and Compensation}

As stated by Remark \ref{rmk1}, fractional Doppler shift estimation should be prioritized and performed first. For LOS channels, a reasonable solution is to compensate the received signal for Doppler frequency. From (6), the compensated received signal is given by:
\begin{equation}	
	\label{doppler_compensate}  
	r_{\kappa_{\text{loc}}} [n]  =r[n]  e^{\imath 2\pi \frac{\kappa _{\text{loc}}}{N}n},
\end{equation}
where $\kappa_{\text{loc}}$ denotes the local compensation frequency. After compensation, the normalized Doppler frequency is \(K = k + \kappa - \kappa_{\text{loc}}\). When \(\kappa_{\text{loc} }=\kappa\), \(K\) reduces to its integer part \(k\), which means that no fractional-Doppler-induced energy leakage remains in \(\Upsilon(m,m')\). In this scenario, the compensated received signal $y_{\kappa_{\text{loc}}}[m]$ after DAFT achieves its sparsest representation, and the minimum-entropy principle can thus be adopted to select the proper Doppler compensation \cite{coifmanEntropybasedAlgorithmsBest1992}. To facilitate the subsequent entropy calculation, we define the energy-based probability mass function in terms of the normalized discrete signal energy, given by:
\begin{equation}	
	\label{energy of noise}  
		p_{\kappa_{\text{loc}}}(m) = \frac{|y_{\kappa_{\text{loc}}}(m)|^2}{\sum_{m=0}^{N-1} |y_{\kappa_{\text{loc}}}(m)|^2}, \quad m = 0, 1, \ldots, N-1,
\end{equation}

The estimated fractional Doppler shift $\hat{\kappa}$ is obtained as:
\begin{equation}	
	\label{est_kappa}  
 \hat{\kappa} = \argmin_{\kappa_{\text{loc}} \in (0, 1)} H(\kappa_{\text{loc}})=\argmin_{\kappa_{\text{loc}} \in (0, 1)} \left( -\sum_{m=0}^{N-1} p_{\kappa_{\text{loc}}}(m) \ln p_{\kappa_{\text{loc}}}(m) \right).
\end{equation}

\begin{figure}[!h]
	\centering
	\includegraphics[width=3.5in]{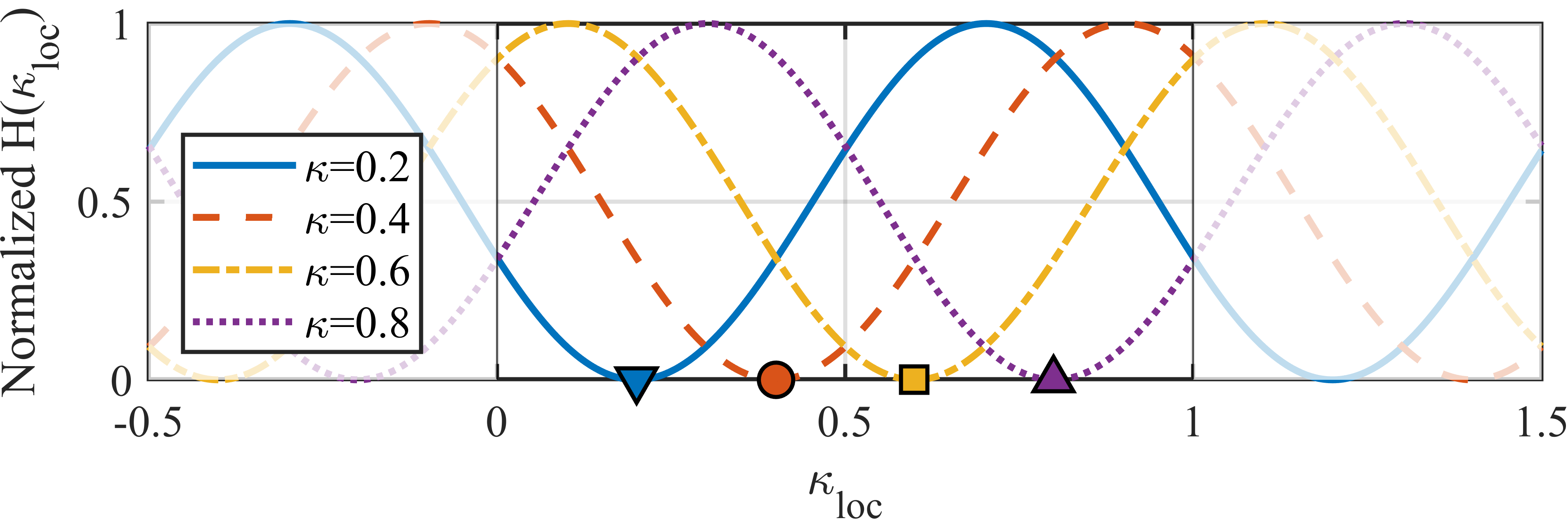}
	\caption{Illustration of local entropy valleys under modulo-1 arithmetic for \(\kappa\in\{0.2,0.4,0.6,0.8\}\).}
	\label{entropy}
	\vspace{-10pt}
\end{figure}

To accelerate the solution of (16), we adopt a two-stage search strategy. 
A coarse search is first performed over $\kappa_{\mathrm{loc}}\in(0,1)$ to 
provide a conservative initialization by locating a 
promising entropy valley, within which Fibonacci search is then applied for 
local refinement\cite{liMatchedFilteringbasedChannel2026a}. 

\begin{rmk}\label{rmk2}
As observed from Fig.~\ref{entropy}, $H(\kappa_{\mathrm{loc}})$ shows an 
approximately unit-periodic and regular profile. Although it is not strictly unimodal over \((0,1)\), the entropy minimum associated with the true fractional Doppler shift typically lies in a local valley. Hence, $N_{\mathrm{c}}$ uniformly spaced coarse samples 
are first evaluated over $(0,1)$, and the coarse minimizer is obtained as
\begin{equation}
	i^\star=\arg\min_i H(\kappa_{\mathrm{loc}}^{(i)}).
\end{equation}

A local refinement interval is then initialized as
\begin{equation}
[\kappa_{\mathrm{L}},\kappa_{\mathrm{U}}]
=
[\kappa_{\mathrm{loc}}^{(i^\star)}-\Delta_{\mathrm{c}},
\kappa_{\mathrm{loc}}^{(i^\star)}+\Delta_{\mathrm{c}}],
\end{equation}
where $\Delta_{\mathrm{c}}$ is chosen according to the coarse-grid spacing. In practice, it can be simply set as 
$\Delta_{\mathrm{c}}=1/N_{\mathrm{c}}$, which covers one coarse grid spacing 
around the selected coarse minimum and provides satisfactory performance in the considered scenarios.  Within this interval, Fibonacci search is applied for refinement. Let $F_j$ 
denote the $j$-th Fibonacci number. The two interior points are given by
\begin{equation}
\kappa_1=\kappa_{\mathrm{L}}
+\frac{F_{j-2}}{F_j}
(\kappa_{\mathrm{U}}-\kappa_{\mathrm{L}}), 
\kappa_2=\kappa_{\mathrm{L}}
+\frac{F_{j-1}}{F_j}
(\kappa_{\mathrm{U}}-\kappa_{\mathrm{L}}).
\end{equation}

At each iteration, the subinterval that cannot contain the smaller entropy value is discarded until the stopping criterion is met. In this way, Fibonacci search is used as a 
local refinement tool after coarse localization, without requiring strict 
global unimodality of $H(\kappa_{\mathrm{loc}})$ over $(0,1)$.
\end{rmk}
\subsubsection{Closed-form Estimation for Fractional Delay}
Consider the signal after ideal fractional Doppler compensation. In this case, we have \(K=k\), and \(\Upsilon(m,m')\) equals \(C\) only when \(C\iota+\Delta=pC\) for \(p\in\mathbb{Z}\), while approaching zero at other grid points. Solving this yields
\begin{equation}
	\frac{\Delta}{C}=p-\iota,
\end{equation}
 Substituting into $\varTheta \left( m,m^{\prime} \right)$, at the non-zero valid sampling points, the absolute amplitude of the signal can be expressed as
\begin{equation}
	\label{A_p}
	A_p = N \cdot |\mathrm{sinc}(\iota - p)|.
\end{equation}

Since the fractional delay satisfies \(\iota\in(-0.5,0.5)\), the global peak appears at \(p=0\). Let \(m_0\) denote the corresponding value of \(m\), with the amplitude given by
\begin{equation}
	\label{A_0}
	A_0 = N \cdot |\mathrm{sinc}(\iota)|.
\end{equation}

The amplitudes at the sampling points offset by one unit from the peak, i.e., \(|p|=1\), are given by
\begin{equation}
	\label{A_1}
	A_{1} = N \cdot |\mathrm{sinc}(1 - \iota)|,
\end{equation}
\begin{equation}
	\label{A_-1}
	A_{-1} = N \cdot |\mathrm{sinc}(1 + \iota)|.
\end{equation}

When $\iota > 0$, the secondary peak appears at $p=1$ (i.e., $m=m_0-C$). At this point, the main peak and the secondary peak satisfy an analytical relationship:
\begin{equation}
	\label{iota_A}
	\frac{A_{1}}{A_0 + A_{1}} = \left|\frac{ \frac{\sin(\pi \iota)}{\pi (1-\iota)} }{ \frac{\sin(\pi \iota)}{\pi \iota} + \frac{\sin(\pi \iota)}{\pi (1-\iota)} }\right| = \left|\frac{ \frac{1}{1-\iota} }{ \frac{1}{\iota} + \frac{1}{1-\iota} }\right| = \iota.
\end{equation}

The case for $\iota < 0$ is analogous. Summarizing, we obtain:
\begin{equation}
	\label{iota}
	\hat{\iota}=\begin{cases}
		\frac{A_{1}}{A_{0}+A_{1}}, & A_{1}>A_{-1},\\
		-\frac{A_{-1}}{A_{0}+A_{-1}}, & A_{-1}>A_{1}.
	\end{cases}
\end{equation}
\begin{rmk}
In implementations under noisy conditions, a direct comparison between $A_1$ and $A_{-1}$ may be unreliable when $|\iota|$ is close to zero. Therefore, a dead-zone-based sign-decision rule is adopted: for a predefined small threshold $\eta_0 \in (0,1)$, if $\eta = |A_1 - A_{-1}|/A_0 < \eta_0$, the received signal is correlated with two reference patterns, as in an MF-like estimator, corresponding to positive and negative fractional-delay hypotheses, respectively, and the sign of \(\hat{\iota}\) is determined by the larger correlation value. This operation is used only for sign disambiguation, while the magnitude of $\hat{\iota}$ is still obtained from \eqref{iota}, with negligible additional computational complexity.
\end{rmk}

\subsection{Complexity Analysis}
The computational complexity of the proposed algorithm is analyzed in two parts. 
For fractional Doppler estimation, the coarse search with \(N_c\) uniformly spaced samples requires \(N_c\) entropy evaluations. The subsequent Fibonacci search requires approximately \(I=\lceil \log_{\phi}(L_0/\epsilon_{\rm prec})\rceil\) iterations, where 
$\phi=(1+\sqrt{5})/2$ is the golden ratio, $\epsilon_{\mathrm{prec}}$ is the 
desired precision, and $L_0=2\Delta_{\mathrm{c}}$ denotes the initial refinement 
interval length after coarse localization. Each evaluation involves (i) Doppler 
compensation with $N+n_{\mathrm{cp}}$ complex multiplications, (ii) DAFT 
demodulation with one size-$N$ FFT and $2N$ complex multiplications, and 
(iii) entropy computation with $\mathcal{O}(N)$ real operations. For fractional delay estimation, the 
closed-form ratio method requires only $\mathcal{O}(N+2\delta N\log N)$ 
operations, where $\delta$ denotes the probability that the algorithm enters 
the dead zone. Therefore, the overall complexity is mainly dominated by 
approximately $(N_{\mathrm{c}}+2I+2\delta)$ size-$N$ FFT-equivalent operations, 
yielding $\mathcal{O}((N_{\mathrm{c}}+2I+2\delta)N\log N)$.
\section{Simulation Results}
This section validates the proposed algorithm via Monte Carlo simulations, with the main parameters summarized in Table~\ref{tab:sim_params}. The NTN-TDL-C model~\cite{3gpp_tr38811} is adopted for the LEO-ICAN downlink and implemented in the equivalent baseband domain using fractional-delay filters~\cite{tseFundamentalsWirelessCommunication2005}, where a random normalized residual timing offset \(L_{\rm res}\) accounts for synchronization imperfections. The baselines include MF~\cite{bemaniIntegratedSensingCommunications2024}, MF with generalized Fibonacci search (MF-GFS)~\cite{liMatchedFilteringbasedChannel2026a}, off-grid SBL(OG-SBL)~\cite{yangAFDMOffgridChannel2024}, SPA~\cite{yinPilotAidedChannel2022}, and the root Cram\'er--Rao lower bound (RCRLB)~\cite{zhangAFDMenabledIntegratedSensing2025}. For each SNR, the RMSE over \(E=1000\) independent trials is computed as
\begin{equation}
\operatorname{RMSE}(x)=
\sqrt{\frac{1}{E}\sum_{i=1}^{E}
	\left|\hat{x}^{(i)}-x^{(i)}\right|^2},
\quad x\in\{K,L\}.
\end{equation}

For the model validation in Fig.~7(b), each trial randomly draws a delay--Doppler pair as $L \sim \mathcal{U}[0,l_{\max}], K \sim \mathcal{U}[-k_{\max},k_{\max}].$
Given the same transmitted AFDM frame $\mathbf{x}$, the normalized envelope mismatch averaged over trials, denoted by NEMA, is defined as
\begin{equation}
	{\rm NEMA}(H_{\rm test},H_{\rm ref})
	=
	\frac{1}{E}\sum_{i=1}^{E}
	\frac{
		\left\|
		\left|\mathbf{H}_{{\rm eff},{\rm test}}^{(i)}\mathbf{x}\right|
		-
		\left| \mathbf{H}_{{\rm eff},{\rm ref}}^{(i)}\mathbf{x}\right|
		\right\|_2^2
	}{
		\left\|
		\left| \mathbf{H}_{{\rm eff},{\rm ref}}^{(i)}\mathbf{x}\right|
		\right\|_2^2
	}.
\end{equation}

\begin{table}[!h]
	\centering
	\caption{Simulation Parameters}
	\label{tab:sim_params}
	\footnotesize
	\setlength{\tabcolsep}{0.3pt}
	\renewcommand{\arraystretch}{1.05}
	\setlength{\arrayrulewidth}{0.4pt}
	\begin{tabular}{|c|c|c|c|}
		\hline
		\textbf{Parameter} & \textbf{Value} 
		& \textbf{Parameter} & \textbf{Value} \\
		\hline
		Subcarrier Number 
		& $N=256$ 
		& Subcarrier Spacing 
		& $\Delta f=30$ kHz \\
		\hline
		Carrier Frequency 
		& $f_c=2$ GHz 
		& Search Tolerance 
		& $\epsilon_{\text{prec}}=10^{-5}$ \\
		\hline
		Chirp Parameter 2 
		& $c_2=\sqrt{2}$ 
		& Max Normalized Delay 
		& $l_{\text{max}}=3$ \\
		\hline
		Satellite Altitude 
		& $h=600$ km 
		& Max Normalized Doppler 
		& $k_{\text{max}}=3$ \\
		\hline
		Vehicle Speed 
		& $v=60$ km/h 
		& Residual Timing Offset 
		& $L_{\text{res}}\sim\mathcal{U}[0,l_{\text{max}}]$ \\
		\hline
		Channel Model 
		& NTN-TDL-C 
		& Elevation Angle 
		& $\theta\sim\mathcal{U}(30^{\circ},70^{\circ})$ \\
		\hline
		Dead-zone Threshold
		& $\eta_0 = 0.2$
		& Monte Carlo Trials
		& $E=1000$ \\
		\hline
		Pilot Amplitude
		& $x_{\text{pilot}} = 100$
		& Coarse Search Samples
		& $N_c=4$ \\
		\hline
	\end{tabular}
\end{table}
\begin{figure}[!h]
	\centering
	\includegraphics[width=3.5in]{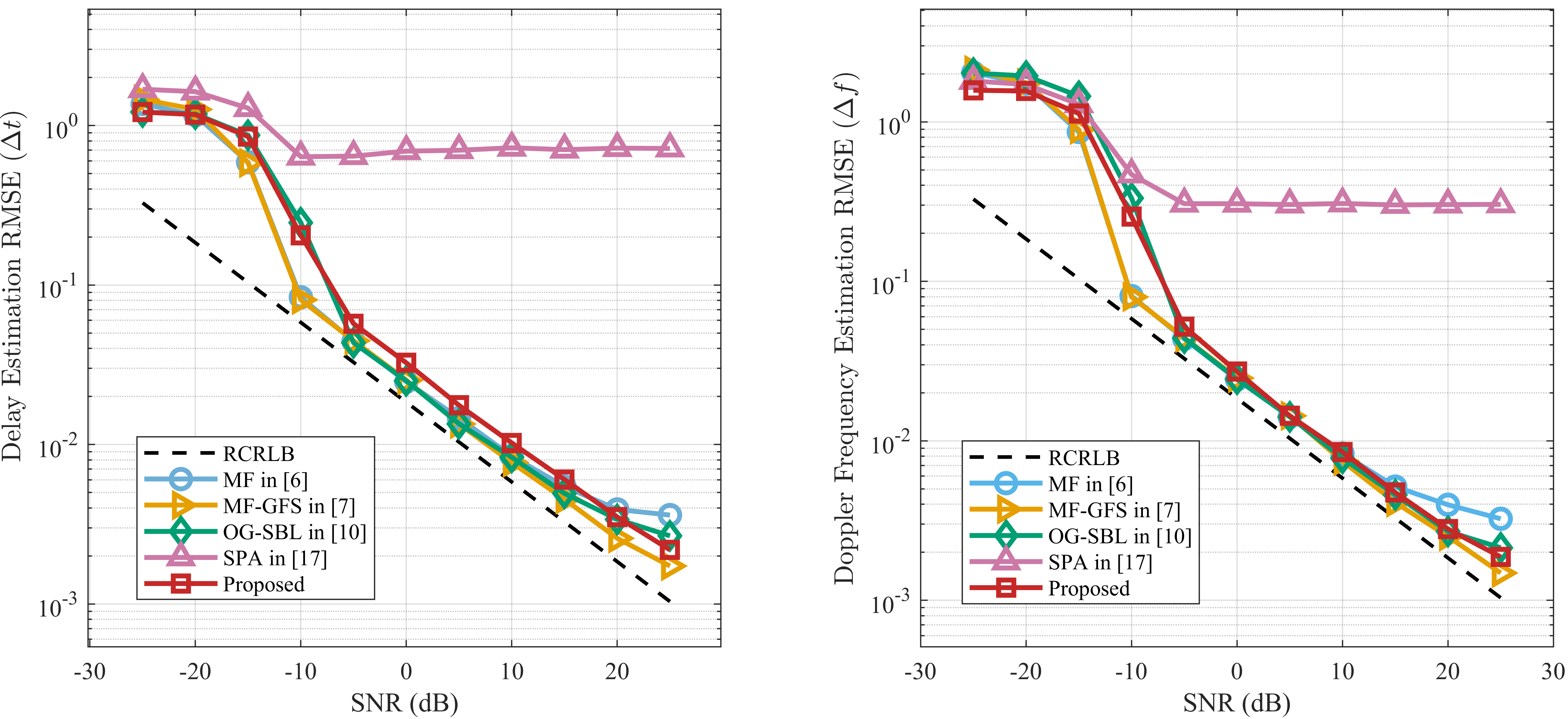}
	\caption{RMSE performance comparison for delay and Doppler frequency estimation.}
	\label{RMSE}
	\vspace{-10pt}
\end{figure}

\begin{figure}[!h]
	\centering
	\includegraphics[width=3.5in]{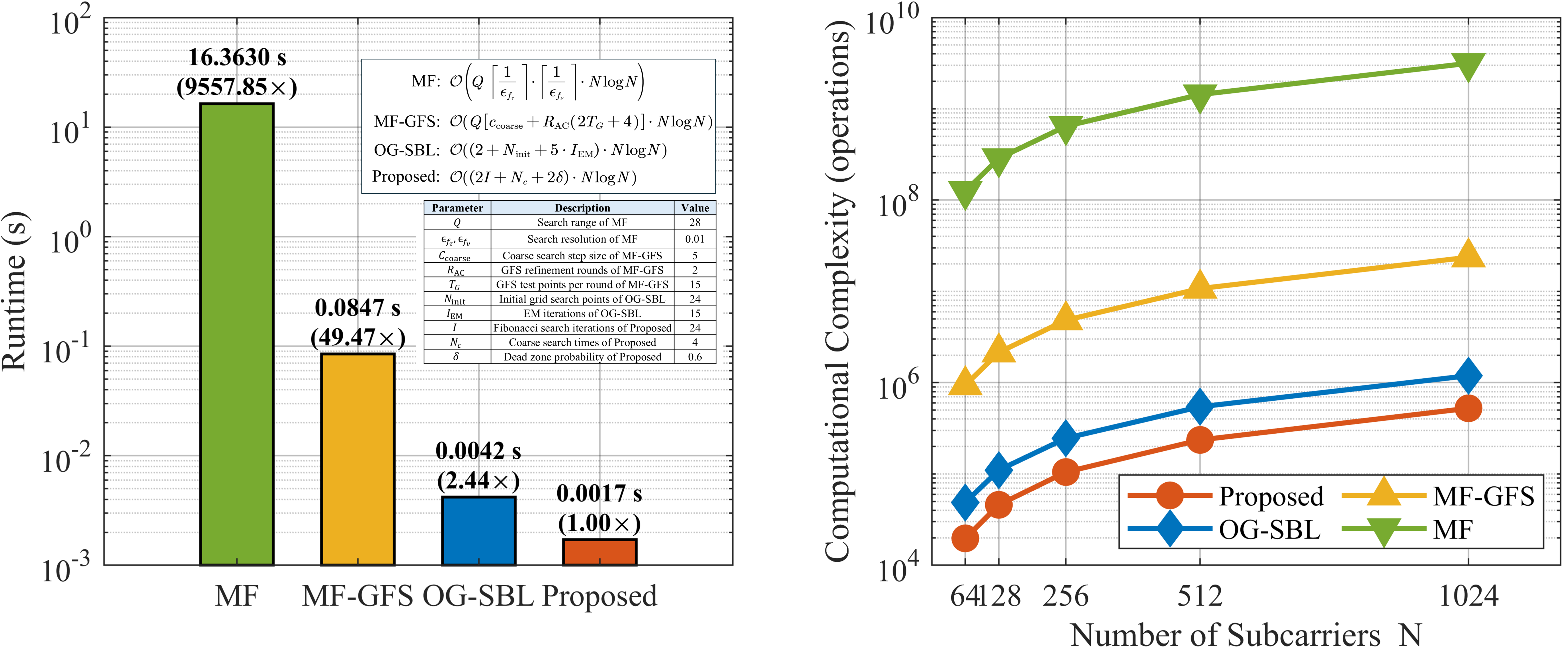}
	\caption{Comparison of runtime and computational complexity across different algorithms.}
	\label{complexity}
	\vspace{-10pt}
\end{figure}

\begin{figure}[!h]
	\centering
	\includegraphics[width=3.5in]{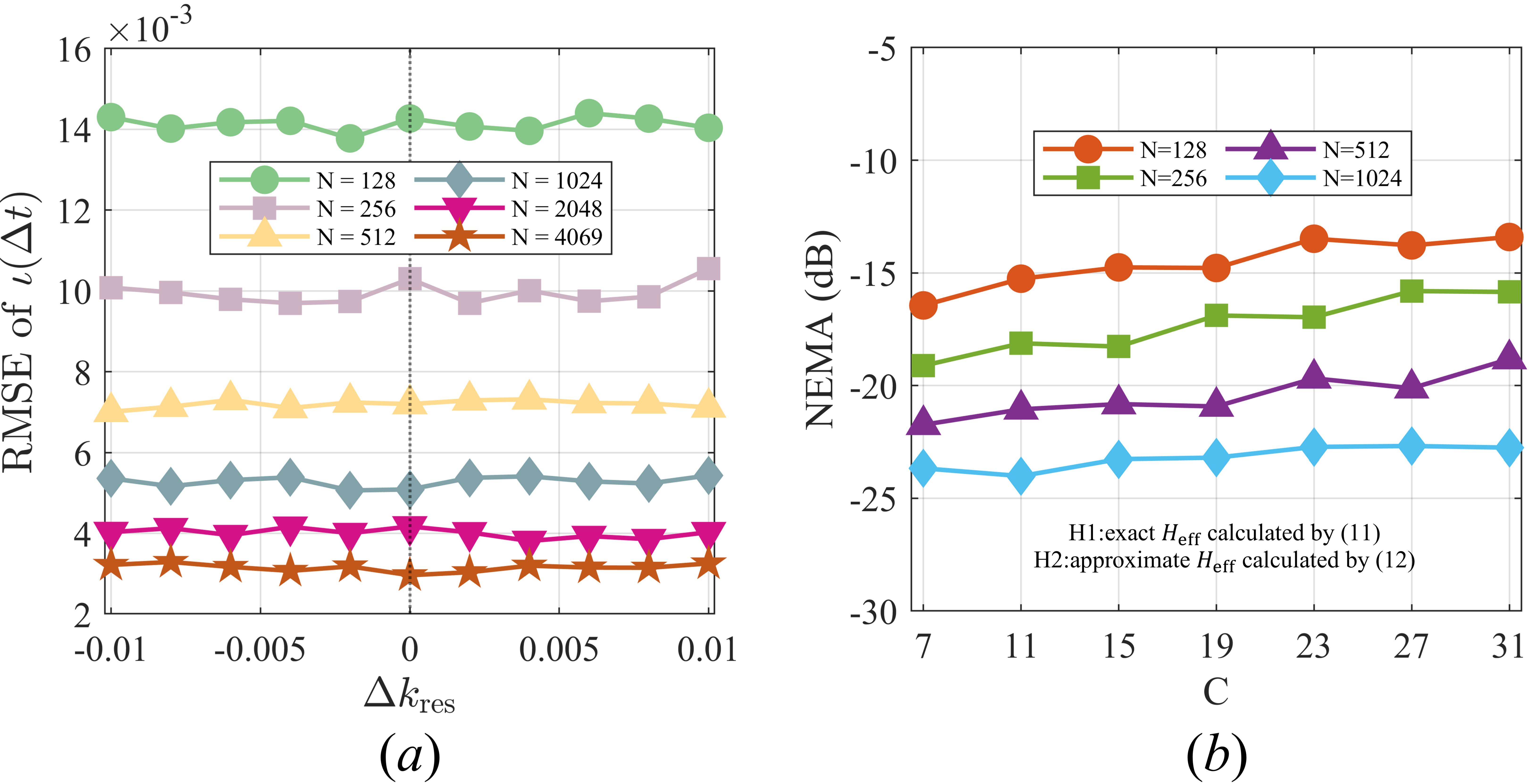}
	\caption{Robustness analysis: (a) Delay-estimation RMSE versus residual Doppler-offset error for different values of $N$ (under $\mathrm{SNR}=10$~dB); (b) NEMA of the proposed approximation (H2) versus the exact reference (H1) over \(C\) and \(N\).}
	\label{Robustness}
	\vspace{-10pt}
\end{figure}

Fig.~\ref{RMSE} compares the RMSE performance of delay and Doppler estimation versus SNR. The proposed method achieves comparable accuracy to benchmark schemes (MF, MF-GFS, OG-SBL), with RMSE decreasing steadily as SNR increases and approaching the RCRLB at high SNR. Unlike SPA, which exhibits an error floor, the proposed method maintains continuous improvement across the entire SNR range. At 25 dB, the delay RMSE reaches approximately \(0.002\Delta t\), and the Doppler RMSE reaches approximately \(0.0019\Delta f\), closely matching the best benchmarks. Fig.~\ref{complexity} shows that at $N = 256$, the proposed method requires only $0.0017~\mathrm{s}$ in a Python-based numerical implementation, substantially faster than MF, MF-GFS, and OG-SBL, while avoiding the SPA error floor. The right panel confirms its $\mathcal{O}(N\log N)$ scaling.

Fig.~\ref{Robustness} summarizes the robustness and model-validation results. In Fig.~\ref{Robustness}(a), the delay estimator exhibits strong robustness against residual Doppler error $\Delta k_{\rm res}$, with the RMSE remaining low and steady over the considered range. Notably, the RMSE decreases consistently as $N$ increases, indicating improved delay-estimation accuracy with larger $N$. Fig.~\ref{Robustness}(b) validates approximation $\mathrm{H2}$ in~(12) by evaluating the NEMA between $\mathrm{H2}$ and the exact reference $\mathrm{H1}$ in~(11). The NEMA  stays within $-24$ to $-13$~dB for different $C$ and $N$, and generally decreases as $N$ grows, confirming the accuracy of the approximate envelope structure.

\section{Conclusion}
This paper has proposed a low-complexity joint fractional delay--Doppler estimator for AFDM-enabled vehicular LEO-ICAN systems. By accounting for spectrum-wrapping, we derived the fractional AFDM input--output relationship and characterized the DAFT-domain envelope structure. We then developed a minimum-entropy-based fractional Doppler compensation method and a closed-form fractional delay estimator. Simulations confirm that the proposed estimator approaches the RCRLB while significantly reducing runtime compared with MF, MF-GFS, and OG-SBL. Future work will focus on severe multipath or NLOS environments.

\appendix

\subsection{Proof of Theorem 1}
\begin{proof}[Proof]
	Let $\Delta\triangleq m'-(m+l^{\mathrm{eq}})$ with $l^{\mathrm{eq}}=(K+CL)_N$. The proof proceeds in four approximation steps, each associated with one condition in Theorem 1.
	
	\emph{Step 1 (Compact-boundary replacement).} Each segment $\mathcal{L}_{m',q}$ has exact boundary $n_q^{\mathrm{start}}=(\lfloor l+\iota+(qN-m')/C\rfloor)_N+1$. Replacing $\lfloor l+\iota+Nq/C-m'/C\rfloor$ by $l+Nq/C+\lfloor\iota-m'/C\rfloor$ perturbs each boundary by $\mathcal{O}(1)$ sample, hence $\mathcal{O}(C)$ samples in total, yielding a relative error $\mathcal{O}(C/N)$. Under the operating regime $l+\iota\in[0,N)$, cyclic folding is triggered only at $q\in\{0,C\}$, and (\ref{F_p}) admits the piecewise form
	\begin{equation}	
		\begin{split} 
			\label{summation_of_F}  
			\mathcal{F}(m,m')\!&\approx\!\!\!\!\!\!\sum_{n=l+1}^{\left(l+\frac{N}{C}+\lfloor\iota-\frac{m'}{C}\rfloor\right)_N}\!\!\!\!\!\!\!\!e^{\imath\frac{2\pi}{N}\Delta n}\\
			&\!+\!\!\sum_{q=1}^{C-1}\!\sum_{n=\left(\frac{Nq}{C}+l+\lfloor\iota-\frac{m'}{C}\rfloor\right)_N+1}^{\left(l+\frac{N(q+1)}{C}+\lfloor\iota-\frac{m'}{C}\rfloor\right)_N}\!\!\!\!\!\!\!\!\!\!\!\!e^{\imath\frac{2\pi}{N}\Delta n}e^{\imath 2\pi\iota q}\\
			&+\!\!\!\!\sum_{n=\left(l+\lfloor\iota-\frac{m'}{C}\rfloor\right)_N+1}^{l}\!\!\!\!\!\!\!\!e^{\imath\frac{2\pi}{N}\Delta n}e^{\imath 2\pi\iota C}.
		\end{split}
	\end{equation}
	
	\emph{Step 2 (Removing the modulo-$N$).} The modulo operation in (\ref{summation_of_F}) is equivalent to a left cyclic shift on $\mathbb{Z}_N$ by $l+1$ positions, so the number of folded samples satisfies
	\begin{equation}
		\#\{\text{folded samples}\}\le l+1\le l_{\max}+1=\mathcal{O}(l_{\max}),
	\end{equation}
	independent of $m'$ and $\iota$. For each folded sample with unfolded index $\hat n=n+N$:
	\begin{equation}
		\label{sum_element}
		e^{\imath\frac{2\pi}{N}\Delta\hat n}=e^{\imath\frac{2\pi}{N}\Delta n}\,e^{-\imath 2\pi\Delta}=e^{\imath\frac{2\pi}{N}\Delta n}\,e^{\imath 2\pi(\kappa+C\iota)}.
	\end{equation}

 Since $|1-e^{\imath 2\pi(\kappa+C\iota)}|\le 2$ over $\mathcal{O}(l_{\max})$ samples, the cyclic-folding error is of order $\mathcal{O}(l_{\max})$ and relatively $\mathcal{O}(l_{\max}/N)$.
	
	\emph{Step 3 (Merging the $q=0,C$ edge segments).} After unfolding, the two spectrum-wrapping-induced edge segments become adjacent but carry distinct phase factors $1$ and $e^{\imath 2\pi C\iota}$. Letting $d\triangleq\lfloor\iota-m'/C\rfloor$, the $q=C$ segment in (\ref{summation_of_F}) spans $n=l+d+1,\dots,l$, whose cardinality is $-d$ when $d<0$ and zero otherwise, i.e.,
	\begin{equation}
		M_C(m')\triangleq\max\{-d,\,0\}=\max\!\left\{-\lfloor\iota-m'/C\rfloor,\,0\right\}.
	\end{equation}
	
	Approximating this segment with the \(q=0\) phase induces an absolute deviation $|e^{\imath 2\pi C\iota}-1|M_C(m')=2|\sin(\pi C\iota)|M_C(m')$, hence relative error $\mathcal{O}(M_C(m')/N)$. For the pilot placement $m'=0$ with $\iota\in(-0.5,0.5)$, $M_C(0)\le 1$, reducing this error to $\mathcal{O}(1/N)$.
	
	Combining Steps 1--3, (\ref{summation_of_F}) reduces to
	\begin{equation}
		\label{simplified_2}
		\mathcal{F}(m,m')\approx\!\sum_{q=0}^{C-1}\!e^{\imath 2\pi\iota q}\!\!\!\!\sum_{n=\frac{Nq}{C}+l+\lfloor\iota-\frac{m'}{C}\rfloor+1}^{l+\frac{N(q+1)}{C}+\lfloor\iota-\frac{m'}{C}\rfloor}\!\!\!\!\!\!e^{\imath\frac{2\pi}{N}\Delta n}.
	\end{equation}
	
	Evaluating the two nested geometric series and using $|1-e^{\imath 2\pi x}|=2|\sin(\pi x)|$ yields
	\begin{equation}
		\label{pre_sinc}
		\left|\mathcal{F}(m,m')\right|\approx\left|\frac{\sin\!\left(\pi(C\iota+\Delta)\right)}{\sin\!\left(\tfrac{\pi}{C}(C\iota+\Delta)\right)}\right|\left|\frac{\sin\!\left(\tfrac{\pi}{C}\Delta\right)}{\sin\!\left(\tfrac{\pi}{N}\Delta\right)}\right|.
	\end{equation}
	
	\emph{Step 4 (Sinc linearization).} In the main-lobe region, $\sin(\pi\Delta/N)=\pi\Delta/N+\mathcal{O}(\Delta^3/N^3)$, so the envelope factor in (\ref{pre_sinc}) becomes $\tfrac{N}{C}|\mathrm{sinc}(\Delta/C)|$ with relative error $\mathcal{O}(\Delta^2/N^2)$, completing the proof of (\ref{magnitude}).
\end{proof}
\bibliographystyle{IEEEtran}
\bibliography{reference}

\end{document}